\documentclass[a4paper,USenglish,cleveref, autoref, thm-restate]{lipics-v2021}
\usepackage[utf8]{inputenc}
\usepackage{xcolor}
\usepackage{url}
\usepackage{amsmath}
\usepackage{xspace}
\usepackage{microtype}
\usepackage{thmtools}

\usepackage{algorithm}
\usepackage{algorithmicx}
\usepackage{algpseudocode}

\algnewcommand{\IIf}[1]{\State\algorithmicif\ #1\ \algorithmicthen}
\algnewcommand{\EndIIf}{\unskip\ \algorithmicend\ \algorithmicif}

\algtext*{EndWhile} 
\algtext*{EndIf} 
\algtext*{EndFor} 

\algdef{SE}[SUBALG]{Indent}{EndIndent}{}{\algorithmicend\ }%
\algtext*{Indent}
\algtext*{EndIndent}

\newcommand{\source}{source\xspace}
\newcommand{\destination}{destination\xspace}
\newcommand{\destinations}{destinations\xspace}
\newcommand{\neigh}{\mathcal{N}}
\newcommand{\Distance}{\mathit{d}}
\newcommand{\Col}{\operatorname{c}}
\newcommand{\ColSet}{\operatorname{cs}}
\renewcommand{\Pr}{\mathbb{P}}
\newcommand{\Layer}{\mathcal{L}}
\newcommand{\In}{\mathrm{in}}
\newcommand{\Out}{\mathrm{out}}
\newcommand{\SP}{\operatorname{SP}}
 
\newcommand{\isactive}{\textsf{active}}
\newcommand{\notactive}{\textsf{non-active}}

\title{Beeping Shortest Paths via Hypergraph Bipartite Decomposition}  

\author{Fabien Dufoulon}{Department of Computer Science, University of Houston, USA}{fabien.dufoulon.cs@gmail.com}{https://orcid.org/0000-0003-2977-4109}{Supported in part by NSF grants CCF-1717075, CCF-1540512, IIS-1633720, and BSF grant 2016419.}
\author{Yuval Emek}{Technion, Israel}{yemek@technion.ac.il}{https://orcid.org/0000-0002-3123-3451}{Supported in part by the Technion Hiroshi Fujiwara Cyber Security Research Center and the Israel National Cyber Directorate.}
\author{Ran Gelles }{Bar-Ilan University, Israel}{ran.gelles@biu.ac.il}{https://orcid.org/0000-0003-3615-3239}{Supported in part by ISF grant No.~1078/17 and BSF grant No.~2020277.}

\authorrunning{F. Dufoulon, Y. Emek, R. Gelles}
\Copyright{Fabien Dufoulon, Yuval Emek, Ran Gelles} 

\keywords{Beeping Networks, Shortest Paths, Hypergraph Bipartite Decomposition}

\ccsdesc{Theory of computation~Distributed algorithms}

\nolinenumbers 
\hideLIPIcs

\EventEditors{Yael Tauman Kalai}
\EventNoEds{1}
\EventLongTitle{14th Innovations in Theoretical Computer Science Conference (ITCS 2023)}
\EventShortTitle{ITCS 2023}
\EventAcronym{ITCS}
\EventYear{2023}
\EventDate{January 10--13, 2023}
\EventLocation{MIT, Cambridge, Massachusetts, USA}
\EventLogo{}
\SeriesVolume{251}
\ArticleNo{28}

\begin{document}

\maketitle

\begin{abstract}
Constructing a shortest path between two network nodes is a fundamental task
in distributed computing.
This work develops schemes for the construction of shortest paths in
randomized beeping networks between a predetermined \emph{\source} node
and an arbitrary set of \emph{\destination} nodes.
Our first scheme constructs a (single) shortest path to an arbitrary
\destination in
$O (D \log\log n + \log^3 n)$
rounds with high probability.
Our second scheme constructs multiple shortest paths, one per each
\destination, in
$O (D \log^2 n + \log^3 n)$
rounds with high probability.

Our schemes are based on a reduction of the above shortest path construction tasks to a decomposition of 
hypergraphs into bipartite hypergraphs:
We develop a beeping procedure that partitions the 
(polynomially-large)
hyperedge set of a hypergraph
$H = (V_H, E_H)$
into
$k = \Theta (\log^2 n)$
disjoint subsets
$F_1 \cup \cdots \cup F_k = E_H$
such that the (sub-)hypergraph
$(V_H, F_i)$
is bipartite in the sense that there exists a vertex subset
$U \subseteq V$
such that
$|U \cap e| = 1$ 
for every
$e \in F_i$.
This procedure turns out to be instrumental in speeding up shortest path constructions under the beeping model.
\end{abstract}

\section{Introduction}
\label{sec:intro}
Constructing a shortest path between given \source and \destination nodes is
among the most fundamental tasks in algorithmic graph theory in general and in
the field of distributed graph algorithms in particular.
In this field, the challenge behind the shortest path task is dominated by the
decentralized nature of the underlying communication network, where in most
distributed computational models, the communication network is assumed to be
identified with (or at least derived from) the input graph (see, e.g.,
\cite{Peleg2000book}).
These computational models differ in various aspects, the most important one
is arguably the communication scheme that determines how information can be
exchanged between the graph nodes.

The distributed computational model that takes the most radical approach in
terms of the communication scheme is the \emph{beeping model}
\cite{CK10, FluryW2010slotted}
that provides an abstraction for networks of extremely simple devices.
In particular, the communication capabilities of each device boil down to two modes of
operation:
\emph{beep}, i.e., emit a signal;
and
\emph{listen}, i.e., detect if a signal is emitted in its (graph) vicinity.
This limited communication scheme means that beeping algorithms may be
implemented in practice by small, cheap, and energy-efficient devices that can
be employed in various applications.
A prominent example is a sensor network, deployed in an airborne manner over a
large terrestrial area, and thus creating an unstructured network composed of a large
number of simple devices.

Despite the extensive literature dedicated to beeping algorithms for various
graph theoretic tasks (see Section~\ref{sec:related-work}), to the best of our knowledge, the
shortest path task has not been studied yet under the beeping model.
The current paper strives to change this situation, advancing the state-of-the-art of the following tasks:
in the simpler case, the goal is to construct a shortest path from a designated \source node to one \destination node.
In the more general case, the goal extends to constructing shortest paths between the \source and multiple \destination
nodes.

Since the structure of the network is unknown and since we do not assume that
the primitive devices have unique IDs, the path construction goal has the
following interpretation:\footnote{%
The algorithms considered in this paper are randomized and the nodes know an
approximation of the graph size, so unique IDs can be  generated.}
Each node holds a binary output and the nodes that output $1$ form a shortest
path $p$ from the \source to the \destination.
We further require that each node along $p$ knows its distance from the
\source.

To be more concrete, consider a network represented as an undirected connected
graph
$G = (V, E)$
over
$n = |V|$
nodes, where one designated node, denoted by
$s \in V$,
is called the \emph{\source}.
Each node is a beeping device and edges correspond to pairs of devices that
can hear each other's beeps.
The execution progresses in discrete rounds so that
(1)
a node that beeps at a certain round cannot listen in the same round;
and
(2)
a node that listens in a round can only distinguish the case that none of its
neighbors beep from the case that some (at least one) neighbors beep.
The nodes are randomized anonymous machines and we assume that they know a
(polynomial) approximation of $n$.\footnote{%
In fact, it suffices for our algorithms that the \source $s$ knows an
approximation of $n$.}

When the execution commences, all nodes are asleep;
in this state, nodes do not perform any computation, nor do they beep.
Some \emph{\destination} nodes, whose set is denoted by
$Y \subseteq V$,
wake up at arbitrary times due to external events (e.g., a triggered sensor).
Other nodes (non-\destinations) wake up once at least one of their neighbors
beeps.
Upon waking up, either due to an external event or hearing a beep, each node
sets its local clock to~$0$ and starts executing a predefined algorithm.
We emphasize that no global clock is assumed, that is, although all nodes
communicate in synchronous rounds, each node has a local definition of ``round
$0$'' (as explained in the sequel, global synchronization is obtained as part
of the algorithm).

In this paper, we develop algorithms for two (families of) tasks related to
the construction of shortest paths from the \source $s$ to the \destinations in $Y$.
To facilitate the exposition, assume for the time being that there is a single
\destination, i.e,
$Y = \{ y \}$,
and consider the task of computing a shortest path $p$ between $s$ and $y$.
Using standard techniques, we can ensure that all nodes are awake and know
their distance from the \source.
Following that, the natural approach would be to construct the desired path
$p$ step-by-step, from $y$ inwards, until $s$ is reached.
Specifically, once we have already constructed a prefix of $p$ that leads from
$y$ to a node $v$ at distance $i$ from $s$, the next step would be to identify
a neighbor $u$ of $v$ whose distance from $s$ is
$i - 1$;
node $u$ is then appointed as the next node along the path $p$ and
the process continues.

The task of identifying such a neighbor $u$ of $v$ under the beeping model
requires some effort as $v$ may have multiple neighbors
$u_{1}, \dots, u_{r}$
at distance
$i - 1$
from $s$, all of which are candidates to form the next node $u$ along~$p$.
The common method to handle this task is to run a competition among the set of
candidates so that each candidate $u_{i}$ picks a random bitstring 
$x_{i} \in \{ 0, 1 \}^{\ell}$
and beeps in a pattern determined by $x_{i}$;
this allows the candidate $u_{i}$ with the lexicographically largest bitstring
$x_{i}$ to win the competition and identify itself as the next node along the
constructed path $p$.

A necessary condition for the success of this method is that the length $\ell$
of the bitstrings, and hence also the length of the beeping patterns, should
be
$\Theta (\log n)$,
leading to a total round complexity of
$\Theta (D \log n)$,
where $D$ denotes the diameter of the communication graph~$G$.
The question that guides our research is whether this round complexity bound
can be improved.
Specifically, can we ``separate'' between the linear dependency on~$D$ and the
(poly)logarithmic dependency on $n$?
We answer this question in the affirmative as stated in the following theorem
(see Theorem~\ref{thm:singlePathConstruction} for the exact statement).

\begin{theorem}[single path, simplified]
\label{thm:main:inf}
Let
$Y \subseteq V$
be a set of \destinations.
There exists an algorithm that constructs a shortest path from the \source
$s$ to an arbitrary \destination
$y \in Y$
in
$O (D \log\log n + \log^{3} n)$
rounds whp.\footnote{%
We say that event $A$ occurs \emph{with high probability}, abbreviated by
\emph{whp}, if
$\Pr(A) \geq 1 - n^{-c}$
for an arbitrarily large constant
$c > 0$.}
\end{theorem}

Next, we consider the task of constructing a shortest path between the \source $s$
and \emph{each} \destination
$y \in Y$.
The naive approach would be to perform the construction implied by
Theorem~\ref{thm:main:inf} separately (non-simultaneously) for each
\destination.
However, the complexity of constructing a path node-by-node would then
increase by a multiplicative factor of~$|Y|$ (at the very least).
This leads to the question of whether the dependency on the number $|Y|$ of
\destinations can be avoided.
We answer this question affirmatively as well, albeit at the cost of
increasing the dependence on $n$ in the $D$-term, as stated in the following
theorem (see Theorem~\ref{thm:MultiplePathConstruction} for the complete
statement).

\begin{theorem}[multiple paths,simplified]
\label{thm:multi:inf}
Let
$Y \subseteq V$
be a set of \destinations.
There exists an algorithm that constructs a shortest-path tree from the
\source $s$ to (all \destinations in) $Y$ in
$O (D \log^{2} n + \log^{3} n)$
rounds whp.
\end{theorem}

\subsection{Our Techniques}
\label{sec:techqniues}

\subparagraph*{The Hypergraph Bipartite Decomposition Problem}
The constructions promised in Theorems~\ref{thm:main:inf} and~\ref{thm:multi:inf} are made possible due to the following combinatorial problem that, on the face of it, does not seem related to shortest path tasks.
We call this problem \emph{hypergraph bipartite decomposition (HBD)} and define it as follows.

A hypergraph
$H = (V_{H}, E_{H})$
is said to be \emph{bipartite} if there exists a vertex subset
$U \subseteq V_{H}$
such that each hyperedge is incident on exactly one vertex of $U$, that is,
$|e \cap U| = 1$
for every
$e \in E_{H}$
(see, e.g., \cite{AharoniK1990extension}).
The goal of the HBD problem is to partition the hyperedge set $E_{H}$ of a
given hypergraph
$H = (V_{H}, E_{H})$
into pairwise disjoint clusters
$F_{1} \cup \cdots \cup F_{k} = E_{H}$
so that the sub-hypergraph
$(V_{H}, F_{i})$
is bipartite for every
$i \in [k]$
with the objective of minimizing the number $k$ of clusters.\footnote{%
We stick to the convention that
$[k] = \{ 1, \dots, k \}$.}\textsuperscript{,}\footnote{%
Note that the HBD notion, as defined in the current paper, is a generalization of the notion of \emph{graph bipartite decomposition}, where the hypergraph
$H = (V_{H}, E_{H})$
is an ordinary graph
(i.e., the rank of each hyperedge
$e \in E_{H}$
is
$|e| = 2$).}


We subsequently encode a solution for the HBD problem on
$H = (V_{H}, E_{H})$
by means of a function
$\Col : E_{H} \rightarrow [k]$
that assigns a \emph{color} $\Col(e)$ to each hyperedge
$e \in E_{H}$
and a function
$\ColSet : V_{H} \rightarrow 2^{[k]}$
that assigns a (possibly empty) color set
$\ColSet(v)$
to each vertex
$v \in V_{H}$.
The functions $\Col(\cdot)$ and $\ColSet(\cdot)$ are subject to the constraint
that for each hyperedge
$e \in E_{H}$,
there exists exactly one incident vertex
$v \in e$
such that
$\Col(e) \in \ColSet(v)$.
Using the aforementioned bipartite sub-hypergraphs terminology, the hyperedges
$e \in E_{H}$
whose color is
$\Col(e) = i$
form the cluster $F_{i}$, whereas the inclusion of color $i$ in the color set
$\ColSet(v)$
of a vertex
$v \in V_{H}$
indicates that $v$ belongs to the vertex subset $U\subseteq V_H$ that realizes the biparticity
of the sub-hypergraph
$(V_{H}, F_{i})$.

The HBD problem plays a pivotal role in the design of our algorithms for the
shortest path construction tasks of Theorems \ref{thm:main:inf} and
\ref{thm:multi:inf}.
To explain this role, let us focus again on the simplified task of
constructing a shortest path $p$ between the \source $s$ and a (single)
\destination $y$.
Recall the aforementioned process that grows the path $p$ from \destination $y$ to \source $s$
step-by-step so that in each step, the path $p$ is extended from a node in
$\Layer_{i}$ to a node in
$\Layer_{i - 1}$,
where $\Layer_{d}$ denotes the set of nodes at distance $d$ from $s$, a.k.a.\
\emph{layer $d$}.
As discussed earlier, assuming that we have already constructed a prefix of
$p$ that leads from $y$ to a node
$v \in \Layer_{i}$
and taking 
$u_{1}, \dots, u_{r}$
to be the neighbors of $v$ in
$\Layer_{i - 1}$,
the bottleneck of the process is the task of identifying one such neighbor
$u_{i}$,
$i \in [r]$.
Indeed, the naive method for this task takes
$\Theta (\log n)$
rounds, which results in
$\Theta (D \log n)$
rounds for the entire process.

This is where the HBD problem comes into play:
In a preprocessing phase, we construct an HBD solution for the
\emph{$(i, i - 1)$-layered hypergraph}
$H = (V_{H}, E_{H})$
defined by setting
$V_{H} = \Layer_{i - 1}$
and
$E_{H} = \{ \neigh_{v} \cap \Layer_{i - 1} \}_{v \in \Layer_{i}}$;
that is, we identify the vertices of~$H$ with the nodes in~$\Layer_{i - 1}$
and introduce a hyperedge for each node
$v \in \Layer_{i}$
that includes all neighbors of $v$ in~$\Layer_{i - 1}$.
In particular, this HBD solution assigns a color
$\Col(v) \in [k]$
to each node
$v \in \Layer_{i}$
and a color set
$\ColSet(u) \in 2^{[k]}$
to each node
$u \in \Layer_{i - 1}$.
The HBD solution for the
$(i, i - 1)$-layered
hypergraph is constructed for each layer
$1 \leq i \leq d_{\max}$,
where
$d_{\max} = \max \{ i \geq 0 \mid \Layer_{i} \neq \emptyset \}$
is the \source's eccentricity in $G$.

The crux of our technique is that once the path growing process reaches node
$v \in \Layer_{i}$,
it is sufficient that $v$ beeps a pattern that encodes $v$'s color
$\Col(v) \in [k]$.
Indeed, the definition of the HBD problem ensures that there is a
\emph{unique} neighbor $u$ of $v$ in
$\Layer_{i - 1}$
with
$\Col(v) \in \ColSet(u)$;
this neighbor $u$ is selected as the next node along the path $p$.
Our solution for the HBD problem uses a polylogarithmic palette size $k$ which
means that the beep pattern that
encodes $\Col(v)$ is of length
$O (\log k) = O (\log\log n)$,
thus providing an exponential improvement over the naive approach and
ultimately leading to a process that grows the entire shortest path $p$ in
$O (D \log\log n)$
rounds.
(In the multiple path task, node $v$ uses a unary encoding of $\Col(v)$ in the
pattern it beeps, hence the larger dependency on $n$.)
Put differently, the colors assigned by the HBD solution serve as ``succinct
locally unique identifiers'' that facilitate an efficient recognition of the
next node in the shortest path $p$.
The key ingredient of our algorithms is cast in the following theorem.

\begin{theorem}[HBD procedure, simplified]
\label{thm:main-ingredient}
There exists a beeping procedure that given a layer
$1 \leq i \leq d_{\max}$,
constructs an HBD solution for the
$(i, i - 1)$-layered
hypergraph with palette size
$k = O (\log^{2} n)$
in
$O (\log^{3} n)$
rounds whp.
In fact, the procedure can be invoked concurrently for all layers
$1 \leq i \leq d_{\max}$,
constructing the corresponding $d_{\max}$ HBD solutions in
$O (\log^{3} n)$
rounds whp.
\end{theorem}

For a hypergraph
$H = (V_{H}, E_{H})$,
we refer to the sum
$|V_{H}| + |E_{H}|$
as the \emph{combinatorial size} of $H$.
As any hypergraph of combinatorial size (at most) $n$ can be ``embedded'' as an
$(i, i - 1)$-layered
hypergraph of some $n$-node graph $G$, Theorem~\ref{thm:main-ingredient} 
means that any such hypergraph can be decomposed into
$k = O (\log^{2} n)$
bipartite hypergraphs.
This turns out to be (combinatorially) tight:
In \cite{ABLP91}, Alon, Bar-Noy, Linial, and Peleg consider the task of broadcasting a message in \emph{radio networks}.
In this setting, if two neighbors of some node communicate at the same round, then that node hears nothing.
Alon et al.\ prove that there exists a bipartite graph $G^{*}$, with $N$ nodes on the left-hand side and $\operatorname{poly}(N)$ on the right-hand side, in which
$\Omega(\log^{2} N)$
rounds are required to broadcast a message $m$, assuming that $m$ is initially known to all nodes on the the left-hand side of $G^*$.

The construction of \cite{ABLP91} implies an
$\Omega(\log^{2} N)$
lower bound for the HBD task on hypergraphs with $N$ vertices and $\operatorname{poly}(N)$ hyperedges.
Indeed, define the hypergraph $H^{*}$ over the left-hand side nodes of $G^{*}$, introducing a hyperedge $e_{u}$ for each right-hand side node $u$ of $G^{*}$ that consists of the left-hand side neighbors of $u$ in $G^{*}$.
If $H^{*}$ can be decomposed into $k$ bipartite hypergraphs $H_1,\ldots,H_k$, then there is a way to broadcast $m$ in the radio-network $G^{*}$ in $k$ rounds:
the subset $U_i$ (left-hand nodes) that realizes the biparticity of~$H_i$ transmits in round~$i$.
Since each hyperedge $e_u\in H_i$ intersects only a single node in~$U_i$, the right-hand node~$u$ learns the message~$m$ in round~$i$. 
Therefore, the HBD problem on $H^{*}$ requires a palette of size
$k = \Omega(\log^{2} N)$.

Theorem~\ref{thm:main-ingredient} and the lower bound of~\cite{ABLP91} yield the following corollary.\footnote{%
We learned about the lower bound~\cite{ABLP91} and the relation between broadcast in radio networks and the HBD task after a preliminary version of this work was made available online.
In hindsight, a $k$~round broadcast scheme in a bipartite radio network implies a $k$-color decomposition of the induced hypergraph.
In particular, the algorithm of Bar-Yehuda, Goldreich, and Itai~\cite{BGI91} gives a broadcast scheme for bipartite radio networks in
$O(\log^{2} n)$
rounds.
Their scheme takes a decay approach (slightly different from the decay approach we adopt in Algorithm~\ref{alg:HBD-abstract}) that guarantees a round with a single transmitter for each receiving node. 
If, similar to our coloring method in Algorithm~\ref{alg:HBD-abstract}, one assigns transmitters with colors according to their transmission rounds and assigns each receiver a color that matches the first round in which a message was received, one obtains an HBD decomposition of the hypergraph defined by the transmitters as nodes and receivers as hyperedges.
It is important to point out that the algorithm of \cite{BGI91} is designed to work under the radio network model and implementing it under the (seemingly harsher) beeping model would require additional machinery, see Section~\ref{sec:preprocessing}.} 

\begin{corollary}
\label{cor:existential-hbd-upper-bound}
Any hypergraph of combinatorial size $n$ can be decomposed into
$k = O (\log^{2} n)$
bipartite hypergraphs and this is existentially tight.\footnote{%
The upper bound promised in Corollary~\ref{cor:existential-hbd-upper-bound} can, in fact, be refined so that the HBD palette size needed for an arbitrary hypergraph
$H = (V_{H}, E_{H})$
is
$k = O \left( \log |V_{H}| \cdot \log |E_{H}| \right)$.}
\end{corollary}




\subparagraph*{A Framework for the Construction of Shortest Paths}
Towards the design of the algorithms promised in Theorems \ref{thm:main:inf} and \ref{thm:multi:inf},
we develop a  framework that supports solving various different shortest path related tasks.
Our framework is composed of three phases:
a \emph{wake-up phase};
a \emph{preprocessing phase};
and
a \emph{path construction phase}.
The wake-up and preprocessing phases are executed ``offline'' for the purpose of constructing the infrastructure that subsequently enables running the third (path construction) phase more efficiently.
In a bird's-eye view, the three phases are as follows (see Section~\ref{sec:statementAndOverview} for an extended overview):
\begin{enumerate}[(i)]

\item 
The goal of the wake-up phase is to wake-up all nodes and set up the parameters required to solve the HBD problem.
In particular, using relatively standard techniques, we guarantee that upon completion of the wake-up phase, each node
$v \in V$
knows to which layer $\Layer_{i}$ it belongs as well as the \source's eccentricity $d_{\max}$.
This phase further guarantees that the nodes obtain globally-synchronized clocks (despite their unsynchronized wake-up).
We elaborate on the wake-up phase in Section~\ref{sec:wakeup}.

\item
The preprocessing phase is dedicated to the solution of the HBD instances defined over the
$(i, i - 1)$-layered
hypergraphs and constitutes the heart of our algorithm(s).
Upon completion of this phase, each node
$v \in \Layer_{i}$
holds the $\Col(v)$ and $\ColSet(v)$ variables that facilitate expedited shortest path growing processes.
To serve different tasks in the path construction phase, we actually compute two sets of these variables, one that supports growing shortest paths \emph{inwards}, from the \destinations to the \source (corresponding to the HBD solution for the
$(i, i - 1)$-layered
hypergraphs as discussed previously), and one that supports growing shortest paths \emph{outwards}, from the \source to the \destinations (corresponding to the HBD solution for the
$(i - 1, i)$-layered
hypergraphs as discussed in the sequel).
We elaborate on the preprocessing phase in Section~\ref{sec:preprocessing}.

\item
The path construction phase is responsible for solving specific tasks related to the construction of shortest paths between the \source and the \destinations based on the infrastructure built in the preprocessing phase.
In this paper, we focus on the task of constructing a (single) shortest path to one \destination (a generalization of Theorem~\ref{thm:main:inf}) and the task of constructing a shortest-path tree to multiple \destinations (a generalization of Theorem~\ref{thm:multi:inf});
we elaborate on those in Sections \ref{sec:single} and \ref{sec:multiple}, respectively.
\end{enumerate}
The infrastructure constructed during the preprocessing phase (based on the wake-up phase) seems fairly useful and we believe that it may be applicable for efficient solutions of other shortest path related tasks.


\subsection{Additional Related Work}
\label{sec:related-work}
The beeping model, defined by Cornejo and Kuhn~\cite{CK10} (c.f.~\cite{FluryW2010slotted}), has been extensively explored in the literature.
Previous work considered solving various tasks in beeping networks, most notably
local symmetry breaking tasks (i.e., MIS and coloring)~\cite{AABHBB11,AABCHK13,JSX16,HL16,CMT17,BBDK18,CMRZ19},
global symmetry breaking tasks (i.e., leader election)~\cite{GH13,FSW14,DBB18},
and
broadcasting of messages~\cite{GH13,HP15,HP16,CD19bc,BBDD19}.
\looseness=-1

To the best of our knowledge, no prior work in the beeping model dealt with the fundamental tasks of constructing shortest paths.
The works~\cite{BBDK18,AGL20} show how to simulate any CONGEST algorithm over a beeping network by first performing a 2-hop coloring and then communicating in each neighborhood using a time-division technique (TDMA) based on the given colors, so that transmissions do not interfere. 
These simulations, however, blow up the round complexity by an $\Omega(\Delta)$ factor to support the TDMA; $\Delta$ is the maximal degree in~$G$. 
Therefore, even though constructing a shortest path in the CONGEST model takes $\Theta(D)$~rounds, a direct simulation of shortest path construction in the beeping model greatly surpasses the complexity we obtain in the current paper.  
We are also not familiar with any works on the algorithmic aspects of the HBD problem, particularly in distributed settings.

Most of the beeping literature so far has assumed synchronous networks where all nodes begin their execution in the same ``round 0''.
Some works (e.g., \cite{AABCHK13,GH13,HP15}) constructed schemes that lift this assumption, allowing some nodes to wake up spontaneously, whereas other nodes wake up when their neighbors beep.
The works~\cite{CK10,DBB20} considered the harsher setting where nodes wake up in arbitrary, adversarially-chosen rounds, and do not necessarily wake up when their neighbors beep.
The works \cite{AABCHK13,FSW14,DBB18,HMP20} 
allow simulating beeping algorithms that assume globally-synchronized clocks over beeping networks with unsynchronized wake-up. Under the assumptions that nodes wake up upon the reception of a beep, one can show that the local clock of any two neighbors $(u,v)\in E$ differ by at most one. Then, each round of the algorithm that assumes globally-synchronized clocks, can be simulated via three rounds: the middle round of every such triplet is overlapping with the corresponding triplet of each neighbor. We shall use this technique in our wake-up phase in Section~\ref{sec:wakeup}.

%
%



\section{Preliminaries}
\label{sec:prelim}
\subparagraph*{Notations}
Given a graph $G=(V,E)$ and a node $v\in V$, we let $\neigh_v =\{u \mid (u,v)\in E \}$ be the neighborhood of~$v$. For any two nodes $u,v\in V$, we let $d(u,v)$ denote their distance, i.e., the length of the shortest path connecting $u$ and~$v$. All logarithms are taken to the base 2 unless otherwise mentioned. 

\subparagraph*{Model and Setting}
We assume the standard beeping model~\cite{CK10}, where a network of $n$ devices is abstracted as the connected graph $G=(V,E)$, where each device is identified with a node $v\in V$.
The computation proceeds in synchronous rounds, where in each round each $v\in V$ can either \emph{beep} or \emph{listen} (but not both). If a node~$v$ listens in a given round it will hear a \emph{beep} only if one of the nodes in its neighborhood $\neigh_v$ beeps in this round; otherwise, $v$ hears a \emph{silence}. The node~$v$ cannot distinguish the case where multiple of its neighbors beep from the case where only a single neighbor beeps. If $v$ beeps in a given round, it cannot tell whether any of its neighboring nodes also beep in that same round.

At the onset, all devices are in a \emph{sleeping} state. 
If some node beeps at some round, all its sleeping neighbors wake up, set their local round-counter to~$0$, and start executing their predefined algorithm (e.g., they start to listen or beep in the following round). Note that the nodes do not possess a global clock and the notion of round number is local. However, for analysis purposes, we denote by \emph{global round} the round number according to the \emph{first} node that wakes up.

Nodes are assumed to know a bound on~$n=|V|$.
We assume that each node has access to an (independent) random string. We do not assume the nodes have IDs, however, IDs can be generated using the random string, with high probability. The nodes receive no inputs. Yet, one distinguished node $s\in V$ is \emph{the \source} and all other nodes know that they are not the \source.
We refer to the set of nodes at distance $i$ from the \source{} $s$ as \emph{(distance) layer $i$}, denoted by
$\Layer_{i} = \{ v \in V \mid \Distance(s, v) = i \}$.

\subparagraph*{Beep Waves} Beep waves~\cite{GH13, CD19bc} are a basic communication primitive in the beeping model, which allows a single node to broadcast a message throughout the network, utilizing the other nodes as relays.
Informally, if some node $v$ starts a beep wave, it beeps in some round~$r$, which causes its neighbors to beep in round~$r+1$, its distance-2 neighbors to beep in round $r+2$, and so on, until the beep wave reaches the entire network. Broadcasting a message~$m$ is achieved by initiating $|m|$ consecutive\footnote{To avoid interference, a node listens in some round~$r$, relays its beep in round $r+1$ and ignores round $r+2$, since its neighbors are relaying the beep then. The node listens again in round $r+3$ which belongs to the following beep wave.} beep waves, where the $i$-th wave behaves as above if $m_i=1$ or is silent (no beeps at all) if~$m_i=0$.

Beep waves can also signal events in the network. In this case, all nodes are set to listen, and $v$ initiates a single beep wave when some event happens. Then, all the nodes learn about that event when the beep wave propagates to them.
Note that once the beep wave has passed some node~$u$, no further beeps that stem from the same beep wave are heard by~$u$. Hence, any future beeps will be correctly interpreted by~$u$ as the next beep wave. Alternatively, $u$ can (locally) terminate or transition to a new state once it hears the beep wave that indicates termination or transitioning into new state. Thus, we obtain a mechanism for synchronizing the network. 

\subparagraph*{Reverse Beep Waves} Reverse beep waves, unlike normal beep waves, only propagate inwards to the \source (or more generally, some designated root). The primitive assumes nodes are globally-synchronized, have already computed their distance to the \source and know some 
bound on the \source's eccentricity, $d_{\max}$ .
%
Nodes separate rounds from 1 to $3d_{\max}$ into triplets of (three consecutive) rounds, referred to as rounds 0,~1 and 2 of that triplet. 
Consider some node $v \in \Layer_i$ for some $i \in [d_{\max}]$. Within each triplet~$r$, the node~$v$ can initiate a reverse beep wave due to an external event. Otherwise (in triplet~$r$), $v$ either relays a reverse beep wave or listens for a reverse beep wave. 
In the latter case, if $v$ hears a  beep
(i.e., issued by its neighbors in $\Layer_{i+1}$)
it
relays the beep in triplet $r+1$, which will then be relayed by $\Layer_{i-1}$, and so on.
Similar to ``standard'' beep waves, one can initiate several consecutive reverse beep waves (separated by three triplets at every node, to avoid interference) which will move separately one layer inwards per triplet without any cross-interference.

It remains to describe how nodes relay reverse beep waves.
If a node $v \in \Layer_i$ initiates 
a beep wave in triplet $r$, then $v$ beeps in round $i \;\mathrm{mod}\; 3 $ of the triplet $r$ and stays idle for the other two rounds of the triplet. 
Otherwise, $v$ listens during triplet~$r$, namely, in all three rounds of the triplet. 
We say that $v$ ``hears'' a reverse beep wave if it hears a beep in round $(i+1) \;\mathrm{mod}\; 3$ of triplet~$r$. Note that only nodes in $\Layer_{i+1}$, among neighbors of $v$, beep in round $(i+1) \;\mathrm{mod}\; 3$ of any triplet, since any neighbor of $v$ must belong in $\Layer_{i-1} \cup \Layer_{i} \cup \Layer_{i+1}$.
To relay the beep heard in triplet~$r$, the node $v\in \Layer_{i}$ beeps in round $i \;\mathrm{mod}\; 3 $ of the triplet $r+1$.

\section{Problem Statement and Algorithm Overview}
\label{sec:statementAndOverview}
Assume an arbitrary network
$G = (V, E)$
in which one node
$s \in V$
is designated as the \source.
At the onset, all nodes are asleep. 
The nodes in a set
$Y = \{ y_{1}, \dots, y_{t} \} \subseteq V$,
referred to as the \emph{\destinations}, experience external events that wake them up (unless already awake), possibly at different times.
We assume that each node
$y_{i} \in Y$
knows that it is a \destination even if it wakes up upon hearing a beep (rather than due to an external event). 
This assumption prevents the edge-case in which some nodes learn that they are \destinations only in late stages of the computation, e.g., after the construction terminates.
We further assume that $y_{1}$ wakes up first and $y_{t}$ last;
this assumption is without loss of generality as the \destinations do not necessarily know whether they are $y_{1}$, $y_{t}$, or neither.

In this paper, we focus on two computational tasks related to the construction of shortest paths between $s$ and $Y$.
Both tasks are specified by means of a \emph{distance policy}
$\pi : 2^{[n]} \rightarrow 2^{[n]}$
that gets the set
$I = \{ 1 \leq i \leq n \mid Y \cap \Layer_{i} \neq \emptyset \}$
of \emph{\destination occupied layer indices} and returns a 
set
$J = \pi(I) \subseteq I$
of \emph{target layer indices}.
For example, the set $J$ can include the minimum \destination occupied layer index or all prime \destination occupied layer indices or all the \destinations, etc.
Note that we think of the distance policy $\pi$ as a private information of the \source and the rest of the nodes can be oblivious to it.
Let
\[
Y_J = \bigcup_{j \in J} Y \cap \Layer_{j}
\]
be the set of all \destinations that belong to the target layers.

In the \emph{single shortest path construction} task, the goal is to construct a shortest path between the \source $s$ and an arbitrary \destination in~$Y_J$. 
%
The output of the single shortest path construction task is provided by means of a binary variable
$z_{v} \in \{ 0, 1 \}$
returned by each node
$v \in V$
when the execution terminates.
These variables are subject to the following constraint:
there exists a shortest path~$p$ between $s$ and an arbitrary \destination in $Y_J$ such that
$z_{v} = 1$
if and only if $v$ is included in~$p$.

The second task we solve is the \emph{shortest-path tree construction} task.
Here, we wish to construct $|Y_J|$ shortest paths, one between $s$ and each one of the \destinations in $Y_J$, so that the union of the shortest paths forms a tree, namely, a \emph{$Y_J$-spanning shortest-path tree rooted at~$s$}.
The output of the shortest-path tree construction task is also provided by means of a binary variable
$z_{v} \in \{ 0, 1 \}$
returned by each node
$v \in V$
when the execution terminates.
These variables are subject to the following constraint:
there exists a $Y_J$-spanning shortest-path tree~$T$ rooted at $s$ such that
$z_{v} = 1$
if and only if $v$ is included in~$T$.

In both the single path and the shortest-path tree construction tasks, the \emph{computation time (complexity)} of the construction is defined to be the number of rounds from the time $y_1$ wakes up until all the nodes return their output and terminate.

\subsection{Scheme Overview}
As discussed in Section~\ref{sec:intro}, our algorithms are composed of three phases, which we briefly describe here.

\paragraph*{The Wake-Up Phase}
The execution commences once $y_1$ wakes up and the objective of the wake-up phase is to wake up the entire network and compute several  parameters required for the next two phases.
Specifically, during the wake-up phase the nodes learn their distance from the \source as well as the \source's eccentricity
$d_{\max} = \max \{ d(v, s) \mid v \in V \}$.
This is performed in a BFS-like manner, where the \source initiates the BFS by issuing a single beep, which is then relayed throughout the network.
The node~$v$ can retrieve its distance from~$s$ from the
round number in which the relayed beep wave is received at~$v$.
Using this mechanism, the nodes also synchronize their clocks so that they all agree on the current round number despite their unsynchronized wake-up (a.k.a.\ having \emph{globally synchronized clocks}).

Following that, we invoke a \emph{distance gathering scheme} that allows all nodes to gather crucial information regarding the \destination set $Y$.
To this end, for each
$1 \leq i \leq d_{\max}$,
let
\[
\SP_{i}
\, = \,
\left\{ v \in V \mid \Distance(s, v) + \Distance(v, y) = i \text{ for some } y \in Y \cap \Layer_{i} \right\}
\]
be the set of nodes that belong to \emph{some} shortest path between~$s$ and $y$ for some (at least one) \destination
$y \in Y \cap \Layer_{i}$.
The distance gathering scheme runs for
$O (D)$
rounds and when it terminates, each node
$v \in V$
knows whether
$v \in \SP_{i}$
for each
$1 \leq i \leq d_{\max}$
(recall that the nodes already know $d_{\max}$).

To implement the distance gathering scheme, a \destination that belongs to layer $\Layer_{i}$ initiates a reverse beep wave in a specific round chosen as a function of $i$ so that the reverse beep waves of \destinations in consecutive layers travel towards the \source in consecutive rounds (up to an agreed upon spacing that prevents interference).
This allows all nodes to identify whether they belong to
$\SP_{i}$
for each
$1 \leq i \leq d_{\max}$,
thus fulfilling the scheme's goal.

Since 
$s \in \SP_{i}$
if and only if
$|Y \cap \Layer_{i}| \neq \emptyset$,
it follows that the distance gathering scheme allows the \source $s$ to learn the set
$I = \{ 1 \leq i \leq n \mid Y \cap \Layer_{i} \neq \emptyset \}$
of \destination occupied layer indices.
The \source can now apply the distance policy $\pi$ to $I$ and obtain the set
$J = \pi(I)$
of target layer indices.
The set $J$ is then disseminated throughout the network in
$O (D)$
rounds by a standard broadcast through beep waves.

Following the dissemination of the set $J$ of target layer indices, each node
$v \in V$
knows whether it is included in the set
\[
\SP_{J}
\, = \,
\bigcup_{j \in J} \SP_{j}
\]
of nodes that belong to a shortest path to some 
\destination in~$Y_J$.
The nodes
$v \in V \setminus \SP_{J}$
set their output bit $z_{v}$ to $0$ and turn themselves off as they do not participate in the next two phases of the algorithm.
The following theorem is established in Section~\ref{sec:wakeup}.

\begin{restatable}{theorem}{wakeupThm}
\label{thm:wakeup}
The wake-up phase takes
$O (D)$
rounds, after which each node
$v \in V$
knows
(1)
$\Distance(v, s)$;
(2)
$d_{\max}$;
(3) $J$;
and
(4)
whether
$v \in \SP_{J}$.
Moreover, all nodes complete the wake-up phase in the same (global) round, thus from that point on, the nodes have globally-synchronized clocks.
\end{restatable}

\paragraph*{The Preprocessing Phase}
The preprocessing phase is invoked by the nodes in~$\SP_{J}$ once the wake-up phase completes.
As discussed in Section \ref{sec:intro}, the objective of the preprocessing phase is to solve the HBD instances defined over the pairs of adjacent layers $\Layer_{i}$ and
$\Layer_{i + 1}$
for each
$0 \leq i \leq d_{\max} - 1$.
For every pair of adjacent layers, we solve two instances of HBD:
one instance with the nodes in
$\Layer_{i}$
playing the role of the vertices and the nodes in
$\Layer_{i + 1}$
playing the role of the hyperedges (the ``inwards'' direction) and one instance with exchanged roles (the ``outwards'' direction). 

The crucial point is that we can parallelize the invocations of the HBD procedure on the different layer pairs as long as they are spaced apart so that interference is avoided.
This results in a round complexity which is independent of the network's diameter~$D$.
The following theorem is established in Section~\ref{sec:preprocessing}.

\begin{restatable}{theorem}{preprocessingThm}
\label{thm:preprocessing}
The preprocessing phase takes
$O (\log^{3} n)$
rounds.
Upon completion of this phase, each node
$v \in \SP_{J}$
holds the variables
$\Col_{\Out}(v), \Col_{\In}(v) \in [k] \cup \{ \bot \}$
and
$\ColSet_{\Out}(v), \ColSet_{\In}(v) \in 2^{[k]}$,
where
$k = O (\log^{2} n)$.
Assuming that
$v \in \Layer_{i}$
for
$0 \leq i \leq d_{\max}$,
the following conditions are satisfied whp:
\begin{enumerate}
    
    \item
    $\Col_{\Out}(v) = \bot$
    if and only if
    $\SP_{J} \cap\, \neigh_{v} \cap \Layer_{i + 1} = \emptyset$,
    i.e., $v$ has no $\SP_{J}$ neighbor in layer $i + 1$;
    
    \item
    $\Col_{\In}(v) = \bot$
    if and only if
    $i = 0$,
    i.e.,
    $v = s$;
    
    \item
    if
    $\Col_{\Out}(v) \in [k]$,
    then there exists exactly one node
    $u \in \SP_{J} \cap\, \neigh_{v} \cap \Layer_{i + 1}$
    such that
    $\Col_{\Out}(v) \in \ColSet_{\Out}(u)$;
    and
    
    \item
    if
    $\Col_{\In}(v) \in [k]$,
    then there exists exactly one node
    $u \in \SP_{J} \cap\, \neigh_{v} \cap \Layer_{i - 1}$
    such that
    $\Col_{\In}(v) \in \ColSet_{\In}(u)$.
    
\end{enumerate}
\end{restatable}

\paragraph*{The Path Construction Phase}
The path construction phase is invoked once the preprocessing phase completes.
As its name implies, this phase is responsible for constructing the actual path(s) and its implementation depends on the specific task we wish to solve.

\subparagraph*{Single Path Construction}
This task is implemented by growing a path step-by-step from the \source $s$ outwards, through the nodes in $\SP_{J}$, until a \destination is reached.
The fact that $\SP_{J}$ 
only contains nodes that reside on a shortest path to some \destination
effectively restricts the above process to sample a single shortest path from $s$ to an (arbitrary) \destination in~$Y_J$, 
out of all such shortest paths.
Employing the $\Col_{\Out}(v)$ and $\ColSet_{\Out}(v)$ variables that each node
$v \in \SP_{J}$
holds, 
the path growing process is executed while spending
$O (\log\log n)$
rounds per hop.
The following theorem is established in Section~\ref{sec:single}.

\begin{restatable}[]{theorem}{SinglePathMainThm}
\label{thm:singlePathConstruction}
Fix some distance policy
$\pi : 2^{[n]} \rightarrow 2^{[n]}$
and let
$Y \subseteq V$
be a set of \destinations.
There exists an algorithm that constructs a shortest path between the \source $s$ and an arbitrary \destination
$y \in Y_J$
in
$O (D \log\log n + \log^{3} n)$
rounds whp.
\end{restatable}

\subparagraph*{Shortest-path Tree Construction}
The naive approach towards implementing this task is to grow a path from each \destination in
$Y_J$
towards the \source.
This approach, however, hits an obstacle if multiple path growing processes interfere with each other (which is inevitable whenever multiple \destinations reside in the same layer).
To be more concrete, think of a scenario where multiple nodes
$v \in \Layer_{i}$
attempt to beep their colors simultaneously, thus introducing interference for the nodes
$u \in \Layer_{i - 1}$
and preventing them from decoding any of the colors.

There are several approaches to bypass this hurdle.
The first is to use superimposed-codes~\cite{KS64}, which allow nodes at layer
$i - 1$
to correctly obtain the colors of up to
$t = |Y|$
layer $i$ nodes that beep simultaneously.
However, beeping a color in this case would require
$O(t^2 \log \log n)$
rounds, which is suitable only when the number $t$ of \destinations is relatively small.
In this paper, we do not make any assumptions on $t$, allowing in particular for
$t = O (n)$
\destinations, which makes this solution too expensive, especially if the value of $t$ is not known to the parties.

Instead, we utilize the fact that our HBD solutions use relatively small color palettes, that is,
$k = O (\log^{2} n)$.
Hence, we can unary encode each color, using $k$ rounds, which allows us to bypass the interference obstacle.
Specifically, given that the tree growing process already reached a node
$v \in \SP_{J} \cap\, \Layer_{i}$,
we allocate $k$ 
(rather than $O (\log k)$)
rounds to the (sub)task of identifying a node
$u \in \SP_{J} \cap\, \neigh_{v} \cap \Layer_{i - 1}$.
Within this window of $k$ rounds, node $v$ beeps in round
$\Col_\In(v) \in [k]$.
Following that, a node
$u \in \SP_{J} \cap\, \Layer_{i - 1}$
identifies itself as part of the constructed tree if it hears a beep in a round indexed by (at least) one of the colors in its color set $\ColSet_{\In}(u)$.
(Notice that several paths that reach layer $i$ may merge into a single path at layer
$i - 1$.)
This process continues, spending
$O (k)$
rounds per layer, until the \source is reached.
The following theorem is established in Section~\ref{sec:multiple}.

\begin{restatable}[]{theorem}{MultiplePathMainThm}
\label{thm:MultiplePathConstruction}
Fix some distance policy
$\pi : 2^{[n]} \rightarrow 2^{[n]}$
and let
$Y \subseteq V$
be a set of \destinations.
There exists an algorithm that constructs a $Y_J$-spanning shortest-path tree rooted at~$s$ in
$O (D \log^{2} n + \log^{3} n)$
rounds whp.
\end{restatable}

\section{Preprocessing Phase}
\label{sec:preprocessing}
The preprocessing phase is invoked, by all $\SP_{J}$ nodes in synchrony, upon completion of the wake-up phase.
Its guarantees are cast in 
Theorem~\ref{thm:preprocessing} stated above.
To avoid cumbersome expressions, the graph theoretic notation used in the scope of the current section takes the subgraph $G(\SP_{J})$ induced by $G$ on $\SP_{J}$ as the underlying graph;
in particular, we use the layer notation $\Layer_{i}$ while we actually refer to
$\Layer_{i} \cap \SP_{J}$.


At the heart of the preprocessing phase, lies an HBD procedure that operates under the beeping model.
This procedure is invoked on an (ordered) layer pair
$(\Layer_{i}, \Layer_{i'})$,
where
$|i - i'| = 1$
(that is, the layers are adjacent), and is implemented (only) by the nodes included in these two layers;
to simplify the presentation, we subsequently denote
$\Layer = \Layer_{i}$
and
$\Layer' = \Layer_{i'}$.
Taking
$M_{v} = \neigh_{v} \cap \Layer'$
to be the neighborhood in $\Layer'$ of a node
$v \in \Layer$,
the goal of the procedure is to construct an HBD solution on the hypergraph
\[
H
\, = \,
\left( \Layer', \{ M_{v} \}_{v \in \Layer \, : \, M_{v} \neq \emptyset} \right)
\, ;
\]
that is, we identify each node
$v \in \Layer$
with the hyperedge~$M_v$.
The HBD procedure assigns a color
$\Col(v) \in [k] \cup \{ \bot \}$
to each node
$v \in \Layer$
and a color set
$\ColSet(u) \in 2^{[k]}$
to each node
$u \in \Layer'$.
These assignments are subject to the following two constraints for each node
$v \in \Layer$:
(1)
$\Col(v) = \bot$
if and only if
$M_{v} = \emptyset$;
and
(2)
if
$M_{v} \neq \emptyset$,
then
there exists exactly one node
$u \in M_{v}$
such that
$\Col(v) \in \ColSet(u)$.
The procedure runs in
$O (\log^{3} n)$
rounds and succeeds whp.

Before presenting the implementation of the HBD procedure, let us explain how the procedure is employed to achieve the objectives of the preprocessing phase as stated in Theorem~\ref{thm:preprocessing}.
First, we invoke the HBD procedure on the layer pair
$(\Layer_{i}, \Layer_{i + 1})$
for
$i = 0, 1, \dots, d_{\max} - 1$,
assigning
$\Col_{\Out}(v) \gets \Col(v)$
for each node
$v \in \Layer_{i}$
and
$\ColSet_{\Out}(u) \gets \ColSet(u)$
for each node
$u \in \Layer_{i + 1}$.
Then, we invoke the HBD procedure on the layer pair
$(\Layer_{i + 1}, \Layer_{i})$
for
$i = 0, 1, \dots, d_{\max} - 1$,
assigning
$\Col_{\In}(v) \gets \Col(v)$
for each node
$v \in \Layer_{i + 1}$
and
$\ColSet_{\In}(u) \gets \ColSet(u)$
for each node
$u \in \Layer_{i}$.
Recalling that each invocation of the procedure takes
$O (\log^{3} n)$
rounds and succeeds whp, the remaining challenge is to parallelize these invocations so that the running time of all of them remains
$O (\log^{3} n)$.

To this end, we exploit the simple fact that the beeps of nodes belonging to layers $i$ and $j$ do not interfere as long as
$|i - j| \geq 3$.
This means that if
$|i - j| \geq 3$,
then we can safely invoke the HBD procedure concurrently on the layer pairs
$(\Layer_{i}, \Layer_{i + 1})$
and
$(\Layer_{j}, \Layer_{j + 1})$
as well as on the layer pairs
$(\Layer_{i + 1}, \Layer_{i})$
and
$(\Layer_{j + 1}, \Layer_{j})$.

\begin{algorithm}[ht]
\caption{\label{alg:preprocessing-phase}%
The preprocessing phase}
\begin{algorithmic}[1]
\For{$\ell \in \{ 0, 1, 2 \}$}

    \State run concurrently for all $0 \leq i \leq d_{\max} - 1$ such that $i = \ell \bmod 3$:
    \Indent
        \State invoke the HBD procedure on the layer pair $(\Layer_{i}, \Layer_{i + 1})$
        \ForAll{$v \in \Layer_{i}$}
            \State $\Col_{\Out}(v) \gets \Col(v)$
        \EndFor
        \ForAll{$u \in \Layer_{i + 1}$}
            \State $\ColSet_{\Out}(u) \gets \ColSet(u)$
        \EndFor
    \EndIndent

    \State run concurrently for all $0 \leq i \leq d_{\max} - 1$ such that $i = \ell \bmod 3$:
    \Indent
        \State invoke the HBD procedure on the layer pair $(\Layer_{i + 1}, \Layer_{i})$
        \ForAll{$v \in \Layer_{i + 1}$}
            \State $\Col_{\In}(v) \gets \Col(v)$
        \EndFor
        \ForAll{$u \in \Layer_{i}$}
            \State $\ColSet_{\In}(u) \gets \ColSet(u)$
        \EndFor
    \EndIndent

\EndFor
\end{algorithmic}
\end{algorithm}

Consequently, the preprocessing phase consists of $6$ subphases, indexed by
$\delta \in \{ \Out, \In \}$
and
$\ell \in \{ 0, 1, 2 \}$.
In subphase
$(\Out, \ell)$
(resp.,
$(\In, \ell)$),
we invoke the HBD procedure on the layer pairs
$(\Layer_{i}, \Layer_{i + 1})$
(resp.,
$(\Layer_{i + 1}, \Layer_{i})$)
for all
$0 \leq i \leq d_{\max} - 1$
such that
$i = \ell \bmod{3}$.
Refer to Algorithm~\ref{alg:preprocessing-phase} for a pseudocode of the preprocessing phase. 
Since each subphase takes
$O (\log^{3} n)$
rounds and succeeds whp, it follows that the whole phase takes
$O (\log^{3} n)$
rounds and succeeds whp as promised in Theorem~\ref{thm:preprocessing}.

\subparagraph*{Implementing the HBD Procedure}
We now turn to explain how the HBD procedure works.
To facilitate the exposition, we first present the procedure as an abstract iterative ``exponential backoff'' process, irrespective of the computational model in which it operates;
refer to Algorithm~\ref{alg:HBD-abstract} for a pseudocode.
Following that, we provide an implementation of this process in the beeping model that runs in
$O (\log^{3} n)$
rounds and succeeds whp.

\begin{algorithm}[ht]
\caption{\label{alg:HBD-abstract}%
An abstract view of the HBD procedure when invoked on a layer pair
$(\Layer, \Layer')$}
\begin{algorithmic}[1]

\ForAll{$v \in \Layer$}
\State $\Col(v) \gets \bot$
\EndFor

\ForAll{$u \in \Layer'$}
\State $\ColSet(u) \gets \emptyset$
\EndFor

\For{$i = 0, 1, \dots, \lfloor \log n \rfloor$}
    \For{$j = 1, \dots, \Theta (\log n)$}
        \ForAll{$u \in \Layer'$}
            \State $X_{u}(i, j) \gets \mathrm{Bernoulli}(2^{-i})$
            \If{$X_{u}(i, j) = 1$}
                \State $\ColSet(u) \gets \ColSet(u) \cup \{ (i, j) \}$
            \EndIf
        \EndFor
        \ForAll{$v \in \Layer$ such that $\Col(v) = \bot$}
            \If{$X_{u}(i, j) = 1$ for exactly one node $u \in M_{v}$}\label{line:HBD-abstract:test-exectly-one}
                \State $\Col(v) \gets (i, j)$
            \EndIf
        \EndFor
    \EndFor
\EndFor

\end{algorithmic}
\end{algorithm}

Consider an invocation of the HBD procedure on a layer pair
$(\Layer, \Layer')$.
Initially, the procedure sets
$\Col(v) \gets \bot$
for each node
$v \in \Layer$
and
$\ColSet(u) \gets \emptyset$
for each node
$u \in \Layer'$.
Following that, the procedure consists of
$\lfloor \log n \rfloor + 1$
successive \emph{epochs}, each running for
$\Theta (\log n)$
iterations, so that iteration
$1 \leq j \leq \Theta (\log n)$
of epoch
$0 \leq i \leq \lfloor \log n \rfloor$
is responsible for the assignment of color
$(i, j)$.

During iteration
$(i, j)$,
each node
$u \in \Layer'$
samples a Bernoulli random variable
$X_{u}(i, j)$
whose success probability is
$2^{-i}$
and adds the color
$(i, j)$
to $\ColSet(u)$, assigning
$\ColSet(u) \gets \ColSet(u) \cup \{ (i, j) \}$,
if
$X_{u}(i, j) = 1$.
Assuming that a node
$v \in \Layer$
is still uncolored (i.e.,
$\Col(v) = \bot$)
at the beginning of iteration
$(i, j)$,
we assign
$\Col(v) \gets (i, j)$
if there is exactly one node
$u \in M_{v}$
such that
$X_{u}(i, j) = 1$.

\begin{lemma}
\label{lem:abstract-HBD}
The abstract HBD procedure (Algorithm~\ref{alg:HBD-abstract}) constructs a feasible HBD solution whp. The procedure takes $O(\log ^2n)$ rounds and features a color palette of size $k=O(\log^2 n)$.
\end{lemma}
\begin{proof}
Observe first that the process is designed so that if
$\Col(v) = (i, j)$
then there exists exactly one node
$u \in M_{v}$
with
$(i, j) \in \ColSet(u)$
(this holds with probability $1$).
To establish the assertion, it remains to show that upon termination of the procedure, the function $\Col(\cdot)$ satisfies
$\Col(v) \neq \bot$
for all nodes
$v \in \{ w \in \Layer \mid M_{w} \neq \emptyset \}$
whp.

To this end, fix some node
$v \in \Layer$
with
$|M_{v}| = r > 0$
and consider epoch
$i = \lfloor \log r \rfloor$
of the procedure.
We say that iteration
$1 \leq j \leq \Theta (\log n)$
of epoch $i$ is successful if
$X_{u}(i, j) = 1$
for exactly one node $u \in M_{v}$.
For each node
$u \in M_{v}$,
let $A_{u}$ denote the event that
$X_{u}(i ,j) = 1$
and
$X_{u'}(i, j) = 0$
for all nodes
$u' \in M_{v} \setminus \{ u \}$.
Observe that
\[
\Pr(A_{u})
\, = \,
2^{-\lfloor \log r \rfloor} \cdot \left( 1 - 2^{-\lfloor \log r \rfloor} \right)^{r - 1}
\, \geq \,
\frac{1}{r} \cdot \left( 1 - \frac{2}{r} \right)^{r-1}
\, ;
\]
it is then easy to verify that
$\Pr(A_u)=1$
for
$r=1$
and
$\Pr(A_u)=1/4$
for
$r=2$;
for
$r \geq 3$,
the inequality yields
$\Pr(A_u) \geq \Omega (1 / r)$.
Since the events $A_{u}$,
$u \in M_{v}$,
are disjoint, it follows that the probability that iteration $j$ is successful is
\[
\Pr \left( \bigvee_{u \in M_{v}} A_{u} \right)
\, = \,
\sum_{u \in M_{v}} \Pr(A_{u})
\, \geq \,
r \cdot \Omega (1 / r)
\, = \,
\Omega (1)
\, \text{,}
\]
independently of the other iterations.
Therefore, at least one iteration of epoch $i$ is successful whp which means that if $v$ is not colored during an earlier epoch, then $v$ is colored by the end of epoch $i$ whp.
The assertion follows by a union bound over all nodes
$v \in \Layer$.
\end{proof}

\smallskip

For the beeping implementation of the HBD procedure, notice that the only component of the aforementioned abstract iterative process that cannot be implemented, out of the box, under the beeping model is the part in which an uncolored node
$v \in \Layer$
tests whether
$X_{u}(i, j) = 1$
for exactly one node
$u \in M_{v}$
(Line~\ref{line:HBD-abstract:test-exectly-one} in Algorithm~\ref{alg:HBD-abstract}).
To overcome this obstacle, we use the following standard technique:
Iteration
$1 \leq j \leq \Theta (\log n)$
of epoch
$0 \leq i \leq \lfloor \log n \rfloor$
is implemented by means of
$\Theta (\log n)$
\emph{sub-iterations}, each consisting of two rounds.
A node
$u \in \Layer'$
that samples
$X_{u}(i, j) = 1$
beeps once in each sub-iteration
$\ell = 1, \dots, \Theta (\log n)$,
where the round in which $u$ beeps is picked uniformly at random (and independently of the the other sub-iterations) among the two rounds of sub-iteration $\ell$.
An (uncolored) node
$v \in \Layer$
assigns
$\Col(v) \gets (i, j)$
if and only if $v$ hears exactly one beep in each sub-iteration during iteration $j$ of epoch $i$.

\begin{lemma}
\label{lem:HBD-beeping-implementation}
A node
$v \in \Layer$
assigns
$\Col(v) \gets (i, j)$
if and only if
$X_{u}(i, j) = 1$
for exactly one node
$u \in M_{v}$
whp.
\end{lemma}
\begin{proof}
The \emph{if} direction clearly holds with probability $1$.
For the \emph{only if} direction, notice that if
$X_{u}(i, j) = 1$
and
$X_{u'}(i, j) = 1$
for two distinct nodes
$u, u' \in M_{v}$,
then $u$ and $u'$ beep in the same round in all sub-iterations of iteration~$j$ of epoch~$i$ in an exponentially decreasing probability with the number of sub-iterations.
This means that $v$ avoids assigning
$\Col(v) \gets (i, j)$
whp. 
\end{proof}

Combining Lemmas~\ref{lem:abstract-HBD} and~\ref{lem:HBD-beeping-implementation}, we conclude that the beeping-model implementation of the HBD procedure succeeds whp.
Theorem~\ref{thm:preprocessing} follows as this implementation consists of
$O (\log n)$
epochs, each consists of
$O (\log n)$
iterations, each consists of
$O (\log n)$
sub-iterations, each consists of
$O (1)$
rounds, which sums up to
$O (\log^{3} n)$
rounds in total.
The palette size $k=O(\log^2 n)$ remains as guaranteed by the abstract HBD procedure, Lemma~\ref{lem:abstract-HBD}.

\section{Single Shortest Path Construction} 
\label{sec:single}

In this section we consider the task of constructing a single path between the \source and an arbitrary \destination. We describe the path construction phase which is invoked by all $\SP_J$ nodes in synchrony, upon completion of the preprocessing phase. The pseudo-code of the path construction phase is provided in Algorithm~\ref{alg:singlePathConstruction}. After the completion of this phase, we are guaranteed with the statement of Theorem~\ref{thm:singlePathConstruction}, which we now recall.

\begin{algorithm}[!ht]
\caption{Path construction phase (for node $v$)}
\label{alg:singlePathConstruction}
\begin{algorithmic}[1]
\State \textbf{Init:}  $state_v \gets \isactive$ if $v = s$; otherwise 
 $state_v \gets \notactive$
\Statex

\For {$\ell = 1$ to~$j_{\max} \triangleq \max J$} 
    \If {$state_v = \isactive$ and $v \notin Y_J$ and $\Distance(v,s)=\ell-1$} \label{line:beepColorCondition}
        \State beep the value~$c_\Out(v)$  \Comment {In $O(\log\log n)$ rounds}
        \label{line:single:online - beep color}
    \ElsIf{$\Distance(v,s)=\ell$}
        \State listen for $O(\log\log n)$ rounds and set $\hat c$ as the value heard
        \If {$\hat c \in cs_\Out(v)$} 
            $state_v \gets \isactive$
        \EndIf
    \Else
        \State stay idle for $O(\log\log n)$ rounds  \Comment{Ensures synchronized iterations}
    \EndIf   
\EndFor

\State output~1 if $state_v = \isactive$; otherwise output~0 
\end{algorithmic}
\end{algorithm}

\SinglePathMainThm*

Nodes in $\SP_J$ construct a shortest path outwards, layer by layer, from the \source to an arbitrary \destination in~$Y_J$.  
More precisely, the path is constructed in (up to) $j_{\max} = \max J$ iterations of $O(\log \log n)$ rounds, such that each iteration extends the path by one layer outwards. Initially, the \source is active. In iteration $1 \leq \ell \leq j_{\max}$, the active node $v_{\ell-1}$, which is in $\Layer_{\ell-1}$, beeps the binary encoding of its color $c_\Out(v_{\ell-1})$. That is, for $1\le i \le O(\log\log n)$, node $v_{\ell-1}$ beeps in the $i$-th round if and only if the $i$-th bit of $c_\Out(v_{\ell-1})$ is $1$. Then, every node $u \in \SP_J \cap\, \Layer_{\ell}$ that hears a color that belongs in its color set $cs_\Out(u)$ becomes active.  Theorem~\ref{thm:preprocessing} guarantees that among $\SP_J \cap\, \neigh_{v_{\ell-1}} \cap\, \Layer_{\ell}$, exactly a single node becomes active. 
%
After at most $j_{\max}$~iterations, the path will reach some \destination in~$Y_J$; this \destination will remain silent, preventing the path from extending any further.
All the nodes that have become active during this process output~$1$ and form the shortest path, whereas all of the other nodes output~$0$. 

\begin{lemma}
\label{lem:thirdPhasePath}
The path construction phase (Algorithm~\ref{alg:singlePathConstruction}) takes $O(D \log \log n)$ rounds. Upon completion, the set of nodes $\{v \in V \mid z_v = 1\}$ forms a shortest path between $s$ and some \destination in $Y_J$. 
\end{lemma}

\begin{proof}
Bounding the round complexity is straightforward: The outer loop runs for $j_{\max}=O(D)$ iterations, where each iteration takes $O(\log\log n)$ rounds. 

Next, we prove by induction on $0 \leq \ell \leq j_{\max}$ that the active nodes at the end of the $\ell$-th iteration form a shortest path between $s$ and either (a) (an arbitrary node in) $\SP_J \cap\, \Layer_\ell$, or (s) some \destination $y \in Y_j$ for $j \in J \cap [\ell]$. Since $\SP_{j_{\max}}$ consists of \destinations only, the induction statement for $\ell = j_{\max}$ implies the correctness of the third subphase. 

For the base case $\ell = 0$ (i.e., the initial state), the only active node is $s$ and since $s \in \SP_{J}$, the claim holds trivially. Next, assume that the above holds for some $0 \leq \ell < j_{\max}$ and consider the next iteration, with $\ell'=\ell+1$. By the induction hypothesis, the active nodes at the end of the $\ell$-th iteration form a shortest path between $s$ and either (a) an arbitrary node in $\SP_J \cap\, \Layer_\ell$, or (s) some \destination $y \in Y_j$ for $j \in J \cap [\ell]$.
In case (s), the path no longer extends since the \destination $y$ does not beep its color $c_\Out(y)$ (Line \ref{line:beepColorCondition}). Hence, we now consider case (a). Let $(s,v_1)(v_1,v_2)\cdots (v_{\ell-1},v_{\ell})$ be the currently-constructed shortest path. 
Recall that all these nodes are active at the end of the $\ell$-th iteration (and beyond). 
Note that $v_{\ell} \in \Layer_{\ell}$. Thus, during iteration $\ell' = \ell + 1$, the node $v_{\ell}$ beeps its color $c_\Out(v_{\ell})$ (Line~\ref{line:single:online - beep color}). Theorem~\ref{thm:preprocessing} ensures that there exists exactly one node $u \in \SP_J \cap\, \neigh_{v_{\ell}} \cap\, \Layer_{\ell'}$ that has $c_\Out(v_{\ell}) \in cs_\Out(u)$. Also note that no other node in $\Layer_{\ell}$ is active (by the induction hypothesis), thus $u$ hears the color $c_\Out(v_{\ell})$ without any interference and becomes active. This satisfies the induction claim for $\ell'$.
\end{proof}

The above lemma immediately leads to proving our first main theorem.

\begin{proof}[Proof of Theorem \ref{thm:singlePathConstruction}]
The round complexity and correctness of the path construction phase both follow from Theorems~\ref{thm:wakeup} and~\ref{thm:preprocessing}, and Lemma~\ref{lem:thirdPhasePath}.
\end{proof}


\section{Shortest-path Tree Construction}
\label{sec:multiple}

In this section we generalize the task of constructing a single shortest path (described in Section~\ref{sec:single}) to the case where we wish to obtain multiple shortest path, one per \destination, that together form a tree rooted at the \source. We describe the tree construction phase which is invoked by all $\SP_J$ nodes in synchrony, upon completion of the preprocessing phase. The pseudo-code of the tree construction phase is depicted in Algorithm~\ref{alg:multi:online}. After the completion of this phase, we are guaranteed with the statement of Theorem~\ref{thm:MultiplePathConstruction}, which we now recall. 

\begin{algorithm}[th]
\caption{Tree construction phase (for node $v$)}
\label{alg:multi:online}
\begin{algorithmic}[1]
\State \textbf{Init:}  $state_v \gets \isactive$ if $v \in Y_J$; otherwise 
 $state_v \gets \notactive$
\Statex

\For {$\ell = 1$ to~$j_{\max}$}
    \If {$state_v = \isactive$ and $d(v,s)=j_{\max} - \ell + 1$}
        \State during the next $k=O(\log^2 n)$ rounds, beep only in the $\Col_\In(v)$-th round  
        \label{line:multi:online - beep color}
    \Else
        \State listen for $k=O(\log^2 n)$ rounds
        \State set $\hat c = \{ i  \mid \text{ a beep was heard in round }i\}$
        \If {$\hat c \cap \ColSet_\In(v) \ne \emptyset$}
            $state_v \gets \isactive$
        \EndIf
    \EndIf   
\EndFor

\State output~1 if $state_v = \isactive$; otherwise output~0. 
\end{algorithmic}
\end{algorithm}

\MultiplePathMainThm*

Towards the above goal,
each \destination in~$Y_J$ initiates the construction of a single path that grows from that \destination towards the \source.
Paths that start at the furthest \destinations in $Y_J$ start growing first, one layer per iteration. In each new iteration, new paths may start growing (if a layer with \destinations in $Y_J$ is reached) but also multiple shortest paths may merge together by reaching the same node.
This process continues until all shortest paths merge together at the \source, resulting in a $Y_J$-spanning shortest-path tree rooted at~$s$.

More precisely, the tree is constructed in $j_{\max} = \max J$ iterations of $O(\log^2 n)$ rounds. In iteration $1 \leq \ell \leq j_{\max} $, any active node $v_a$ in $\Layer_{j_{\max} - \ell +1}$ beeps the unary encoding of its color $c_\In(v_{a})$: that is, node $v_a$ beeps in round $c_\In(v_{a})$ of the iteration. Then, every node $u \in \Layer_{j_{\max} - \ell}$ that hears a color that belongs in its color set $cs_\In(u)$ becomes active. (Note that here, $u$ may hear multiple such colors.) Theorem~\ref{thm:preprocessing} guarantees that for any $v_a \in \Layer_{j_{\max} - \ell +1}$, exactly a single node in $\SP_J \cap\, \neigh_{v_a} \cap \Layer_{j_{\max} - \ell}$ becomes active. 
After $j_{\max}$~iterations, all shortest paths (where some may have merged together) reach the \source.
All the nodes that have become active during this process output~$1$ and form the $Y_J$-spanning shortest-path tree,
whereas all the other nodes output~$0$.

\begin{lemma}
\label{lem:thirdPhaseTree}
The tree construction phase (Algorithm~\ref{alg:multi:online}) takes $O(D \log^2 n)$ rounds. Upon completion, the set of nodes $\{v \in V \mid z_v = 1\}$ forms a $Y_J$-spanning shortest-path tree rooted at $s$.  
\end{lemma}

\begin{proof}
Bounding the round complexity is straightforward: The outer loop runs for $j_{\max}=O(D)$ iterations, where each iteration takes $O(\log^2 n)$ rounds. 

Next, we prove by induction on $0 \leq \ell \leq j_{\max}$ that at the end of iteration $\ell$, the active nodes in $\Layer_{L}$---where $L = \{j_{\max} - \ell,\ldots,j_{\max}\}$ and $\Layer_{L} \triangleq \bigcup_{j\in L}\Layer_j$---form (1) a forest, with roots in $R_\ell \subseteq \SP_J \cap\, \Layer_{j_{\max} - \ell}$, that is also (2) the union of shortest paths from every \destination in $Y_J \cap \Layer_{L}$ to (arbitrary nodes in) $R_\ell$. Importantly, the induction statement for $\ell = j_{\max}$ implies that the set of nodes $\{v \in V \mid z_v = 1\}$ forms a $Y_J$-spanning shortest-path tree rooted at~$s$. 

For the base case $\ell = 0$ (i.e., the initial state), the set of \destinations in $\Layer_{j_{\max}}$ satisfies (1) and (2) trivially. Next, assume that the above holds for some $0 \leq \ell < j_{\max}$ and consider the next iteration, with $\ell'=\ell+1$. By the induction hypothesis, the active nodes in $\Layer_{L}$ form (1) a forest $F_\ell$, with roots $R_\ell \subseteq \SP_J \cap\, \Layer_{j_{\max} - \ell}$, that is also (2)
the union of shortest paths from every \destination in $Y_J \cap \Layer_{L}$ to~%
$R_\ell$.
Consider some root $v_a \in R_\ell$ of that forest. Note that $v_a \in \Layer_{j_{\max} - \ell}$. Hence, $v_a$ is active at the end of the $\ell$-th iteration (and beyond).
Thus, during iteration $\ell' = \ell + 1$, the node $v_{a}$ beeps its color $c_\In(v_{a})$ (Line~\ref{line:multi:online - beep color}). Theorem~\ref{thm:preprocessing} ensures that there exists a single node $u \in \SP_J \cap\, \neigh_{v_{a}} \cap \Layer_{j_{\max} - \ell'}$ that has $c_\In(v_{a}) \in cs_\In(u)$.
%
Note that although other nodes in $\Layer_{\ell}$ may be active, beeping the unary encoding of $c_\In(v_{j_{\max}-\ell})$ ensures that $u$ hears the color $c_\In(v_{j_{\max}-\ell})$ and becomes active.

Let $L' = \{j_{\max} - \ell',\ldots,j_{\max}\}$ and $\Layer_{L'} \triangleq \bigcup_{j\in L'}\Layer_j$. The set of active nodes in $\Layer_{L'} \setminus \Layer_L$ are the new root nodes, denoted by $R_{\ell'}$, and can be separated into (a) the (aforementioned) nodes which become active by hearing a color, and (s) \destinations in $\Layer_{j_{\max} - \ell'}\setminus R_{\ell'}$, which are active from the start, and do not extend an already-constructed path initiated at some \destination in $Y_J\cap \Layer_L$. 
Note that all trees in $\Layer_{L}$ are node-disjoint branches of the first set (a) of active nodes.
Hence, the set of active nodes in $\Layer_{L'}$ forms a forest, rooted in $R_{\ell'}$ that is also the union of (a) for every \destination $y \in Y_J \cap \Layer_{L}$, a shortest path from $y$ to (an arbitrary node in) $R_{\ell'}$ and (s) for every \destination $y \in (Y_J \cap \Layer_{j_{\max} - \ell'})\setminus R_{\ell'}$, if there are any, a 0-length shortest path. 
This satisfies the induction claim for $\ell'$.
\end{proof}

The above lemma immediately leads to proving our second main theorem.
\begin{proof}[Proof of Theorem \ref{thm:MultiplePathConstruction}]
The round complexity and correctness of the tree construction phase both follow from Theorems~\ref{thm:wakeup} and~\ref{thm:preprocessing}, and Lemma~\ref{lem:thirdPhaseTree}.
\end{proof}

\section{Wake-up Phase} 
\label{sec:wakeUp}\label{sec:wakeup}
 
Due to our unsynchronized wake-up assumption, nodes may start the wake-up phase in different (global) rounds. 
Moreover, at the onset of the computation, nodes have very little information about the network's topology. 
The wake-up phase takes care of these issues. That is, upon completion of the phase, we are guaranteed with the statement of Theorem~\ref{thm:wakeup}, which we now recall.

\wakeupThm*

\begin{algorithm}[htp]
\caption{Wake-up phase (for node $v$)}
\label{alg:wakeup}
\begin{algorithmic}[1]
\Statex \textbf{First subphase:}
\State beep in the first round and stay idle for the next round  \Comment{Upon wake up} \label{line:wakeupBeep}

\If{$v = s$} \Comment{$s$ starts the synchronization beep wave.}
    \State beep in round 4 and stay idle in rounds 3,~5 and 6  \label{line:baseBeep}
\Else
    \State listen, starting from round 3, until a beep is heard in some round $r$ \label{line:relayBeep1}
    \State beep in round $r+1$ and stay idle for the next two rounds \label{line:relayBeep2}
\EndIf

\Statex
\Statex \textbf{Second subphase:}
\For{each triplet $i \geq 1$ of consecutive rounds $r_i$,~$r_i+1$,~$r_i+2$} \label{line:startSimulation} \Comment{triplet $i$ simulates step $i$ of $v$}
    \If{$v$ beeps in step~$i$ of  \textsc{EstimateDiameter}} 
        \State beep (only) in round $r_i+1$ 
        \label{line:simulationBeep}

    \Else \Comment{$v$ listens in step $i$ of \textsc{EstimateDiameter}}
        \State listen in rounds $r_i$, $r_i+1$, and $r_i+2$ \label{line:simulateHear1}
        \If{one beep or more were heard} 
            \State record a beep for step~$i$ in \textsc{EstimateDiameter} simulation \label{line:simulateHear2}
        \EndIf
    \EndIf
    \If{the \textsc{EstimateDiameter} simulation terminate at the end of step~$i$}
        \State exit the loop at the end of round $r_i+2$. \Comment{$v$ now knows $d_{\max}$ and $\Distance(s,v)$.}
    \EndIf
\EndFor \label{line:endSimulation}
\State stay idle for $d_{\max} - \Distance(s,v)$ rounds
\label{line:waitToSynch}

\Statex
\Statex \textbf{Third subphase:}
\For{$j=1$ to $d_{\max}$}
    \State $sp_v[j] \gets 0$
\EndFor
\State invoke the reverse beep wave primitive for~$3(4d_{\max}-2)$ rounds, and
\For{triplet $r = 1$ to~$4d_{\max}-2$ during the reverse beep wave primitive}  
    \If{$v \in Y$ and $r = 3(\Distance(v,s)-1)+1$} \label{line:initiateReverseBeepWave}
        \State initiate a reverse beep wave in triplet $r$
        \State $sp_v[\Distance(v,s)] \gets 1$ \label{line:setSPdDestination}
    \Else
        \If{$v$ relays a beep in triplet $r$}  
            \State $sp_v[(r + 2 + \Distance(s,v))/4] \gets 1$  \label{line:setSPd}
        \EndIf
    \EndIf
\EndFor

\Statex
\Statex \textbf{Fourth subphase:}
\If {$v = s$} 
    \State decide on $J\subseteq [d_{\max}]$ 
    \Comment{Via the distance policy~$\pi$}
    \State broadcast $J$ using beep waves       \Comment{$O(D)$ rounds}
\Else
    \State receive $J$ by listening to the beep waves   \Comment{$O(D)$ rounds}
\EndIf
\If {$v\not\in \SP_J$}
    output~$0$ and terminate 
    \label{line:terminateSPhase2}
\EndIf
\end{algorithmic}
\end{algorithm}


\subsection{Detailed Description} 
The wake-up phase consists of four subphases, which we now describe; see Algorithm \ref{alg:wakeup} for the complete pseudo-code.

In the first subphase, every node beeps upon waking up, which in turn wakes any non-awake neighbor of this node. This part guarantees that all nodes wake up within $D$~rounds. 
(Note, however, that $D$ is unknown at this point, and nodes do not count on~$D$ to move to the next subphase.)
Once the \source~$s$ wakes up, it initiates one more beep that is relayed by all the other nodes throughout the network, effectively creating a \emph{synchronization beep wave}: each node that hears this second beep relays it and immediately switches to the second subphase. Therefore, nodes at the same distance to the \source start the second subphase at the same (global) round. (See Section~\ref{sec:prelim} to recall the beep wave primitive description.)


As a corollary of the above, 
neighboring nodes start the second subphase within at most one global round of each other. Under this condition, nodes can simulate, in a simple fashion, any algorithm that assumes nodes have a global clock (Algorithm~\ref{alg:wakeup}; cf.~\cite{FSW14,DBB18,HMP20}). Hence, in the second subphase, nodes simulate a slightly modified version of the \textsc{EstimateDiameter} procedure presented in~\cite{CD19bc}, and presented in Algorithm \ref{alg:estimateDiameter} for completeness. (For clarity, each simulated round of \textsc{EstimateDiameter} is called a \emph{step}.)
The simulation of this (modified) procedure ensures that upon completion, all nodes have learned their distance to the \source~$\Distance(s,v)$ and the eccentricity of the \source~$d_{\max}$. Furthermore, the simulation of \textsc{EstimateDiameter} takes exactly the same number of steps for all nodes. After completing the simulation, each node $v$ waits an additional $d_{\max} - \Distance(s,v)$ rounds before finishing the second subphase, which compensates for its time difference with respect to the time of the \source and guarantees that all the nodes start the third subphase exactly at the same global round. 

Note that from the third subphase onwards, nodes start simultaneously. The third subphase---the distance gathering scheme---computes the sets $\SP_i$  in parallel for every $i \in [d_{\max}]$.
Recall that $\SP_i$ is the set of nodes that reside on some shortest path between $s$ and some \destination in $\Layer_i$.
More precisely, each node $v$ computes a $d_{\max}$-size array $sp_v[\cdot]$ such that, for every $i \in [d_{\max}]$, $sp_v[i] = 1$ if and only if $v \in \SP_i$. 
Computing a single set $\SP_i$, for an a priori known~$i$, is simple. A reverse beep wave (described in Section \ref{sec:prelim}) is initiated by all \destinations in~$Y \cap \Layer_i$ and relayed by all nodes in~$\SP_i$.
Indeed, one can see that the paths along which the reverse beep wave propagates, $\{(y,v_{i-1}),\ldots,(v_1,s) \mid \forall y \in Y\cap\Layer_i, \forall j \in [i-1], v_j \in \Layer_j\}$, are exactly the shortest paths between $s$ and some \destination in $\Layer_i$. 

However, it is more tricky to compute all the sets $\SP_i$ in parallel.
To do so, we pipeline $d_{\max}$ reverse beep waves. The first reverse beep wave is initiated by all the \destinations in~$\Layer_1$ in the first round, the second reverse beep wave by all \destinations in~$\Layer_2$ three triplets later (i.e., in round~10), and so on. (By separating consecutive reverse beep waves by three triplets, reverse beep waves do not interfere with each other.) 
Intermediate nodes deduce 
which reverse beep wave they currently relay
from the triplet in which they hear the beep. In particular, when $v$ relays a beep in triplet~$r$, then $r = 3(i-1) + 1 + i - \Distance(s,v)$ for some $i \in [d_{\max}]$;
the $i$-th beep wave starts in triplet $3(i-1) + 1$ and moves one layer towards the \source per subsequent triplet. Hence, whenever $v$~relays a beep in triplet~$r$, it computes $i = (r + 2 + \Distance(s,v))/4$ and concludes that~$v \in \SP_i$.


In the fourth and final subphase, the \source decides on some set of distances $J \subseteq [d_{\max}]$ by applying the distance policy, $J=\pi(I)$ with $I = \{ 1 \leq i \leq n \mid Y \cap \Layer_{i} \neq \emptyset \}$. We will assume that $Y_J$ is non-empty; otherwise, no node will execute the following phases and thus no path will be constructed (recall that  $Y_J = \bigcup_{j\in J} Y\cap \Layer_j$).
The \source then broadcasts~$J$ to all nodes using $d_{\max}$ consecutive beep waves. 
More precisely, $s$ broadcasts the (indicator function) string $x = x_1 x_2 \ldots x_{d_{\max}}$, where $x_i = 1$ if and only if $i \in J$, from which $J$ can be extracted trivially. Finally, nodes which do not belong in $\SP_J$ turn themselves off and terminate with output~$0$. Consequently, the following phases are executed only by nodes in $\SP_J = \bigcup_{j\in J}\SP_j$.

\subparagraph*{Diameter Estimation} 
For completeness, we give the pseudo-code of the modified version of \textsc{EstimateDiameter} (from~\cite{CD19bc}), used in the second subphase, as well as the corresponding proofs. 

\begin{algorithm}[htp]
\caption{\textsc{EstimateDiameter}~\cite{CD19bc} (for node $v$)}
\label{alg:estimateDiameter}
\begin{algorithmic}[1]
\State $d_v \gets 0 , d_{max} \gets 0$

\If {$v = s$} 
    \State beep in the first step \label{alg:beepFirst}
    \State listen until $v$ does not hear a beep in three consecutive steps: $r-3$, $r-2$ and $r-1$ \label{alg:beepRevWait}
    \State $d_{max} \gets r/3 - 1$ \label{alg:beepRevCompute}
    \State broadcast $d_{max}$ using beep waves
\Else
    \State listen in all steps, until the first step $r'$ in which $v$ hears a beep
    \State $d_v \gets r'$
    \State beep in the next step $r'+1$  \label{alg:beepNext}
    
    \For{each triplet of steps $r'+3k-1$, $r'+3k$, $r'+3k+1$, where $k \geq 1$} \label{alg:beepRev} 
        \State listen in step $r'+3k-1$ \Comment{Listen for a reverse wave} \label{alg:listenRev}
        \State stay idle for step $r'+3k$
        \If{beep heard (during step $r'+3k-1$)}
            \State beep in step $r'+3k+1$ \Comment{Relay two steps later to avoid interference} \label{alg:beepSecond}
        \Else 
            \State stay idle in step $r'+3k+1$ and exit the for loop \label{alg:beepRev2}
        \EndIf
    \EndFor
    \State receive $d_{max}$ using beep waves \label{alg:beepRevReceive}
\EndIf

\State stay idle for $d_{max}-d_v$ steps \Comment{For simultaneous termination}
\State output $(d_v,d_{max})$ and terminate 

\end{algorithmic}
\end{algorithm}

\begin{lemma}
\label{lem:estimateDiameter}
\textsc{EstimateDiameter} assumes that all nodes wake up simultaneously and takes $O(d_{max})$ steps. Upon completion, every node $v$ knows $\Distance(v,s)$ and $d_{\max}$. Moreover, all nodes terminate \textsc{EstimateDiameter} in the same step.
\end{lemma}

\begin{proof}
The \source $s$ beeps in the first step. By a simple induction (and Line~\ref{alg:beepNext} of Algorithm~\ref{alg:estimateDiameter}), all nodes at distance $0 \leq i \leq d_{max}$ (from $s$) hear their first beep in step $i$ and relay this beep in step $i+1$. Hence, for every node $v \in V$, the output $d_v$ is indeed $\Distance(v,s)$ (even for $s$ due to the initialization).

The above mentioned relay beeps start $d_{max}$ consecutive ``reverse beep waves'' (slightly different from those described in Section \ref{sec:prelim}). More precisely, nodes at distance $i$ from $s$ start the $i$th reverse beep wave in step $i+1$ through the relay beeps. To prove this, first consider the following claim: for any node $v \neq s$, after its first beep, node $v$ beeps in step $r = \Distance(v,s)+3k+1$ (for some $k \geq 1$) if and only if one of its neighbors at distance $\Distance(v,s)+1$ from $s$ beeps in step $r-2$. To show the claim, it suffices to show that no other neighbor of $v$ may beep in step $r-2$. (The rest follows from the algorithm description.) This is clearly true for neighbors at distance $\Distance(v,s)$ from $s$ (see Line~\ref{alg:listenRev} of Algorithm~\ref{alg:estimateDiameter}). As for the other neighbors of $v$, they must be at distance $\Distance(v,s)-1$ from $s$ and thus they can only beep in steps $\Distance(v,s)$ and $\Distance(v,s)+3k$ (for $k \geq 1$). This proves the claim.
Moreover, note that $r = \Distance(v,s)+3k+1 = (\Distance(v,s) - 1) +3(k+1)-1$ (for any $k \geq 0$). Combined with the above claim, it is clear that the relay beeps produce reverse beep waves
which take two steps to propagate one layer. In addition, each new beep wave starts one step later and one layer further from the \source. Hence, the \source receives a reverse beep wave every three steps: i.e., in steps $2 + 3k$, for $0 \leq k \leq d_{max}-1$. Lines \ref{alg:beepRevWait}--\ref{alg:beepRevCompute} of Algorithm~\ref{alg:estimateDiameter} ensure that the \source correctly computes 
$r = 2 + 3 (d_{max}-1) + 4 = 3 (d_{max} + 1)$ and $d_{max}$.

Let $r_s$ be the step in which the \source starts broadcasting $d_{max}$. (This is done using $d_{max}$ consecutive beep waves---see the beep waves primitive's description in Section \ref{sec:prelim}.)
Note that in all steps $r'' \geq r_s$, all other nodes are executing Line~\ref{alg:beepRevReceive} of Algorithm~\ref{alg:estimateDiameter}: that is, listening and waiting for the beep wave. Hence, all nodes correctly receive $d_{max}$. 
In fact, each node $v \in V$ terminates the broadcast in step $r_s + 3 (d_{max}+1) + \Distance(v,s)$. (Each node detects a beep wave is the last when it does not hear a beep in the subsequent three steps.) Since node $v$ terminates Algorithm~\ref{alg:estimateDiameter} exactly $d_{max}-d_v = d_{max}-\Distance(v,s)$ steps afterward, all nodes terminate simultaneously. This completes the proof. 
\end{proof}

\subsection{Analysis of the Wake-up Phase} First, we prove Lemmas~\ref{lem:firstSubphaseFirstPhase}--\ref{lem:fourthSubphaseFirstPhase} for the subphases' guarantees. These lemmas lead up to a proof of Theorem~\ref{thm:wakeup}. 

\begin{lemma}
\label{lem:firstSubphaseFirstPhase}
The first subphase takes $O(D)$ rounds. Upon completion, all nodes have woken up. Moreover, for every $i \in [d_{max}]$, every node $v \in \Layer_i$ complete the first subphase simultaneously. Hence, every two neighboring nodes $u$ and $v$ complete within one global round of each other.
\end{lemma}

\begin{proof}
Recall that all nodes beep in their first round (Line \ref{line:wakeupBeep} of Algorithm~\ref{alg:wakeup}). As a result, for any two neighboring nodes $u,v \in V$, the node~$v$ wakes up one round after~$u$ at the latest (or vice versa). Hence, all nodes wake up within the first $D$ global rounds. 

Denote by $r$ the global round in which the \source wakes up. The \source starts the synchronization beep wave in round $r+4$ (Line \ref{line:baseBeep} of Algorithm~\ref{alg:wakeup}). For neighbors of the \source, this beep is the first beep they hear starting from their third local round. Hence, they beep in global round $r+5$ to propagate the synchronization beep wave (Lines \ref{line:relayBeep1}--\ref{line:relayBeep2} of Algorithm~\ref{alg:wakeup}). By a simple induction, one can show that nodes at distance $i$ from the \source propagate this second beep wave in global round $r+4+i$, and thus end the first subphase in global round $r+ 7 + i$ (Lines \ref{line:baseBeep} and \ref{line:relayBeep2} of Algorithm~\ref{alg:wakeup}). 
Note that $r+ 7 + i = O(D)$ and thus every node ends the first subphase within $O(D)$ rounds.

Finally, for any $i \in [d_{max}]$ and any $v \in \Layer_i$, its neighboring nodes in $\neigh_v$ belong to $\Layer_{i-1} \cup \Layer_i \cup \Layer_{i+1}$. Hence, $v$ and any of its neighbors end within one global round of each other.
\end{proof}

\begin{lemma}
\label{lem:secondSubphaseFirstPhase}
The second subphase takes $O(D)$ rounds. Upon completion, every node $v \in V$ knows $\Distance(v,s)$ and $d_{max}$. Moreover, all nodes complete the second subphase in the same round, thus from that point on, the nodes have globally-synchronized clocks.
\end{lemma}

\begin{proof}
During the second subphase, nodes simulate a modified version of \textsc{EstimateDiameter} (see the simulation in Lines \ref{line:startSimulation}--\ref{line:endSimulation} of Algorithm~\ref{alg:wakeup} and \textsc{EstimateDiameter} in Algorithm~\ref{alg:estimateDiameter} ). First, we prove the simulation's correctness via a proof of the following claim: for each step~$i$, some node $v$ hears a beep if and only if at least one neighboring node of $v$ beeps during step~$i$. Recall that each node $v$ simulates a step of \textsc{EstimateDiameter} by using a triplet of consecutive (local) rounds. 
Whenever $v$ beeps during step~$i$ of \textsc{EstimateDiameter}, $v$ beeps in the second local round of triplet~$i$ (Line \ref{line:simulationBeep} of Algorithm~\ref{alg:wakeup}). Since any neighboring node of $u$ starts the second subphase within one global round of $v$ (by Lemma~\ref{lem:firstSubphaseFirstPhase}), this implies that $u$ hears that beep within one of its three rounds in triplet $i$ (if $u$ listens during step~$i$). On the other hand, for the same reason, the triplet of rounds $v$ uses to simulate step~$i$ does not intersect with the triplet of rounds any neighboring node uses to simulate any step~$i' \neq i$. Hence, if no neighboring node of $v$ beeps during step~$i$, then $v$ does not hear any beep from other steps~$i' \neq i$ during its step~$i$ (nor does it hear any beep from step~$i$, clearly). This completes the proof of the claim and the simulation's correctness.

By Lemma \ref{lem:estimateDiameter}, all nodes learn their distance to the \source as well as $d_{\max}$ when the simulation of \textsc{EstimateDiameter} terminates.
Moreover, all nodes terminate the modified procedure in the same step, denoted by $t = O(D)$. 
Since each step is simulated using three rounds, each node simulates the modified procedure for exactly $3t = O(d_{\max})$ rounds.
After which, each node $v$ waits $d_{\max} - \Distance(s,v)$ rounds (Line~\ref{line:waitToSynch} of Algorithm~\ref{alg:wakeup}), which causes all the nodes to synchronize their local clocks. Indeed, every node~$v\in V$ 
completes the second subphase 
in global round $r' = (r + 7 + d(s,v)) + 3t + (d_{\max} - d(s,v)) = r+7+3t +d_{\max}= O(D)$.
\looseness=-1
\end{proof}

\begin{lemma}
\label{lem:thirdSubphaseFirstPhase}
The third subphase takes $O(D)$ rounds. Upon completion, for each $v \in V$ and $i \in [d_{\max}]$, $v$ knows if it is in~$\SP_i$, or more precisely, $sp_v[i] = 1$ if and only if $v \in \SP_i$. 
\end{lemma}

\begin{proof}
The third subphase takes $4d_{\max}-2 = O(D)$ triplets, and each such triplet consists of three rounds. Consequently, the subphase takes $O(D)$ rounds.

Now, we show the correctness of the third subphase. 
First, we provide some observations regarding the reverse beep waves. For every $i \in [d_{\max}]$, let the $i$-th reverse beep wave be the wave initiated by all \destinations in $\Layer_i$ (simultaneously) in triplet $3(i-1)+1$ (Line \ref{line:initiateReverseBeepWave} of Algorithm~\ref{alg:wakeup}). (Note that by Lemma~\ref{lem:secondSubphaseFirstPhase}, all \destinations know which layer they are in.) Then, a simple induction on $i$ shows that the $i$-th reverse beep wave is separated from all other beep waves (and in particular, the $i+1$th beep wave) by three triplets at every node. Hence, no two waves interfere with each other. From now on, for some $i \in [d_{\max}]$, we can consider the $i$-th wave independently from all other waves. 

Next, we prove that (1) any node $v$ that participates in the $i$-th reverse beep wave sets $sp_v[i]$ to 1, and (2) node $v$ participates in the $i$-th reverse beep wave if and only if $v \in \SP_i$. 
For (1), the $i$-th wave is initiated by \destinations in $\Layer_i$ in triplet $3(i-1)+1$ and moves one layer inwards per triplet. Since the first subphase runs for $4 d_{\max} - 2 \geq 3(i-1)+1 + d_{\max}$ rounds, the inwards propagation continues until the wave reaches $s$. Thus, for any $j \in [i]$ and any node $v$ in layer $\Layer_{i-j}$ that relays the wave, $v$ does so in triplet $3(i-1)+1+j$. When that happens, the index computed by $v$ is $(3(i-1)+1+ i - \Distance(s,v) + 2 + \Distance(s,v))/4 = i $ and thus $v$ sets $sp_v[i]$ to~$1$ (Line~\ref{line:setSPd} of Algorithm~\ref{alg:wakeup}). (Note that any \destination in~$\Layer_i$ sets $sp_v[i]$ to~$1$ after initiating the wave---Line~\ref{line:setSPdDestination} of Algorithm~\ref{alg:wakeup}.)

For (2), consider the paths along which the wave propagates: $\{(y,v_{i-1}),\ldots,(v_1,s) \mid \forall y \in Y \cap \Layer_i, \forall j \in [i-1], v_j \in \Layer_j\}$. These are exactly the shortest paths between~$s$ and some \destination in~$\Layer_i$. Indeed, first there cannot be any shorter path between~$s$ and some \destination in~$\Layer_i$, since $\Distance(s,y) = i$ (by definition of $\Layer_i$). Second, there cannot be a shortest path $p = \{(y,v_{i-1}),\ldots,(v_1,s)\}$, between $s$ and $y$ and of length $i$, that does not satisfy $\forall j \in [i-1], v_j \in \Layer_j$. Such a path necessarily has two consecutive nodes in the same layer and hence can be transformed into a path between $s$ and $y$ of length $i-1$, which also contradicts $\Distance(s,y) = i$. \looseness=-1
\end{proof}

\begin{lemma}
\label{lem:fourthSubphaseFirstPhase}
The fourth subphase takes $O(D)$ rounds. Upon completion, each node $v \in V$ knows $J$ and whether $v \in \SP_J$. Nodes in~$\SP_J$ continue to the next phase while the rest of the nodes terminate with output~$0$.
\end{lemma}

\begin{proof}
After the third subphase, Lemma~\ref{lem:thirdSubphaseFirstPhase} guarantees that any $v$  knows for every $i \in [d_{\max}]$, whether $v \in \SP_i$. For $v=s$, learning $s\in \SP_i$ is equivalent to learning whether $\SP_i$ is empty, thus the \source learns the set~$I$. 
The \source then computes $J=\pi(I)$ and broadcasts the string that represents its indicator function in $O(D + d_{\max}) = O(D)$ rounds. 
Thus, all nodes learn the set~$J$. 
Since each node $v\in V$ knows~$J$ and knows for every $i \in [d_{max}]$, whether $v\in \SP_i$, 
each node knows whether it belongs to~$\SP_J$ or not. In the latter case, $v$ turns itself off giving the output~0 (Line~\ref{line:terminateSPhase2} of Algorithm~\ref{alg:wakeup}).
\end{proof}

We can now complete the proof of Theorem~\ref{thm:wakeup}.
\begin{proof}[Proof of Theorem~\ref{thm:wakeup}]
The $O(D)$ upper bound on the round complexity directly follows from Lemmas \ref{lem:firstSubphaseFirstPhase}, \ref{lem:secondSubphaseFirstPhase}, \ref{lem:thirdSubphaseFirstPhase} and \ref{lem:fourthSubphaseFirstPhase}.

By Lemma~\ref{lem:secondSubphaseFirstPhase}, every node $v \in V$ knows $\Distance(v,s)$ and $d_{max}$ at the end of the second subphase, satisfying (1) and (2). Moreover, all nodes complete the second subphase simultaneously. Since each node runs the third and fourth subphases for the same number of rounds, all nodes complete the wake-up phase simultaneously. Finally, by Lemma~\ref{lem:fourthSubphaseFirstPhase}, every node $v$ learns $J$ and whether $v \in \SP_J$, satisfying (3) and (4). This concludes the proof. 
\end{proof}

\section*{Acknowledgments}
We would like to thank Mohsen Ghaffari for pointing out that the work of Alon et al.~\cite{ABLP91} immediately gives a lower bound of $\Omega(\log^2 n)$ to the HBD task.
We would also like to thank Louis Esperet and Carla Groenland for pointing out that the accurate statement of Corollary~\ref{cor:existential-hbd-upper-bound} applies to hypergraphs of combinatorial size $n$, rather than hypergraphs over $n$ vertices.


\bibliographystyle{plainurl}
\bibliography{references}

\end{document}